\documentclass{article}
\pdfoutput=1
\usepackage{arxiv}

\usepackage[utf8]{inputenc} % allow utf-8 input
\usepackage[T1]{fontenc}    % use 8-bit T1 fonts
\usepackage{hyperref}       % hyperlinks
\usepackage{url}            % simple URL typesetting
\usepackage{booktabs}       % professional-quality tables
\usepackage{amsfonts}       % blackboard math symbols
\usepackage{nicefrac}       % compact symbols for 1/2, etc.
\usepackage{microtype}      % microtypography
\usepackage{lipsum}		% Can be removed after putting your text content
\usepackage{graphicx}
\usepackage[round,authoryear]{natbib}
\usepackage{doi}
\usepackage{subfigure}
\newtheorem{myDef}{Definition}
\newtheorem{myTheo}{Theorem}
\newtheorem{mylem}{Lemma}
\usepackage{graphicx}
\usepackage[cmex10]{amsmath}
\newtheorem{proof}{Proof}

\title{Prove Symbolic Regression is NP-hard by Symbol Graph}

\date{} 					% Or removing it

\author{ \href{https://orcid.org/0009-0008-0660-7694}{\includegraphics[scale=0.06]{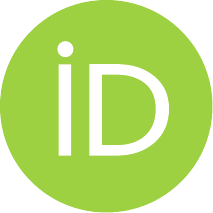}\hspace{1mm}Jinglu Song} \\
	Beijing Key Laboratory of Petroleum Data Mining\\
	China University of Petroleum\\
	Beijing, China \\
	\texttt{jinglusong@hotmail.com} \\
	\And
	\href{https://orcid.org/0000-0001-8217-2305}{\includegraphics[scale=0.06]{orcid.pdf}\hspace{1mm}Qiang Lu\thanks{Corresponding author.}} \\
	Beijing Key Laboratory of Petroleum Data Mining\\
	China University of Petroleum\\
	Beijing, China \\
	\texttt{luqiang@cup.edu.cn} \\
\And
 \href{https://orcid.org/0000-0001-7234-7220}{\includegraphics[scale=0.06]{orcid.pdf}\hspace{1mm}Jingwen Zhang} \\
	Beijing Key Laboratory of Petroleum Data Mining\\
	China University of Petroleum\\
	Beijing, China \\
	\texttt{jingwen.zhang@student.cup.edu.cn} \\
 \And
  \href{https://orcid.org/0009-0009-9694-3313}{\includegraphics[scale=0.06]{orcid.pdf}\hspace{1mm}Bozhou Tian} \\
	Beijing Key Laboratory of Petroleum Data Mining\\
	China University of Petroleum\\
	Beijing, China \\
	\texttt{bozhoutian@gmail.com} \\
 \And
  \href{https://orcid.org/0000-0002-3900-643X}{\includegraphics[scale=0.06]{orcid.pdf}\hspace{1mm}Jake Luo} \\
	Department of Health Informatics and Administration\\
	University of Wisconsin Milwaukee\\
	Milwaukee, United States \\
	\texttt{jakeluo@uwm.edu} \\
 \And
  \href{https://orcid.org/0000-0002-1325-5067}{\includegraphics[scale=0.06]{orcid.pdf}\hspace{1mm}Zhiguang Wang} \\
	Beijing Key Laboratory of Petroleum Data Mining\\
	China University of Petroleum\\
	Beijing, China \\
	\texttt{cwangzg@cup.edu.cn} \\
 }

\begin{document}
\maketitle

\begin{abstract} 
Symbolic regression (SR) is the task of discovering a symbolic expression that fits a given data set from the space of mathematical expressions.  Despite the abundance of research surrounding the SR problem, there's a scarcity of works that confirm its NP-hard nature. Therefore, this paper introduces the concept of a symbol graph as a comprehensive representation of the entire mathematical expression space, effectively illustrating the NP-hard characteristics of the SR problem. Leveraging the symbol graph, we establish a connection between the SR problem and the task of identifying an optimally fitted degree-constrained Steiner Arborescence (DCSAP). The complexity of DCSAP, which is proven to be NP-hard, directly implies the NP-hard nature of the SR problem.

\keywords{Symbolic Regression \and NP-hard \and Degree-Constrained Steiner Arborescence Problem}
\end{abstract}

\section{Introduction}  
Researchers deduce formulas that describe the data law according to their knowledge. For example, Johannes Kepler discovered that Mars’ orbit was an ellipse after he tried to fit the Mars data into various ovoid shapes \citep{udrescu2020ai}. Like the process by which researchers deduce formulas, symbolic regression (SR) tries to discover symbolic expressions fitted by the given data set from the mathematical expression space $\Omega$ \citep{schmidt2009distilling}. To our knowledge, SR has been used in many fields, such as biology \citep{schmidt2011automated}, climate modeling \citep{stanislawska2012modeling}, materials science \citep{wang2019symbolic}, etc. Thus, many different algorithms have been proposed to address the SR problem \citep{he2022taylor,lu2016using,lu2021incorporating}.

Although many kinds of research on the SR problem, there are few works to show that the SR problem is NP-hard \citep{lu2016using,petersendeep,udrescu2020ai,virgolin2022symbolic}. \citeauthor{lu2016using} state that the SR problem is NP-hard but provides no proof. \citeauthor{virgolin2022symbolic} prove it by showing how the decision version of the unbounded subset sum problem can be reduced to a decision version of the SR problem. Although they restrict the mathematical expression space to contain only linear sums of features in the data set, i.e., functions of the form $f(x)=c+\sum_{j=1}^d{x_jm_j}$ with $m_j \in N^+$ and $c\in R$, such as "$x+1$" and "$3x+5$". The symbolic expression is so simple that it cannot represent a complex expression, such as "$sin(x^2+2x^3)$", "$sin(cos(x^2))$", and "$\sqrt{e^x}$". Thus, we need a more perfect method to prove that the SR problem is NP-hard.

In order to overcome the above problem, we introduce the concept of a symbol graph (as illustrated in Figure~\ref{fig:steinergp}). This symbol graph represents the entire mathematical expression space $\Omega$, and we establish a connection between the SR problem and the task of identifying an optimally fitted degree-constrained Steiner Arborescence (DCSAP) \citep{chung1998algorithms,guo2022approximating} within this graph. The complexity of finding a DCSAP, which is proven to be NP-hard in Lemma~\ref{dcsmtnpc}, directly implies the NP-hard nature of the SR problem.
The proposed proof is uniquely relevant to the real-world SR problem as it accounts for more complex expressions within our newly proposed symbol graph.
\begin{figure}[htp]
\centering
\includegraphics[scale=0.55]{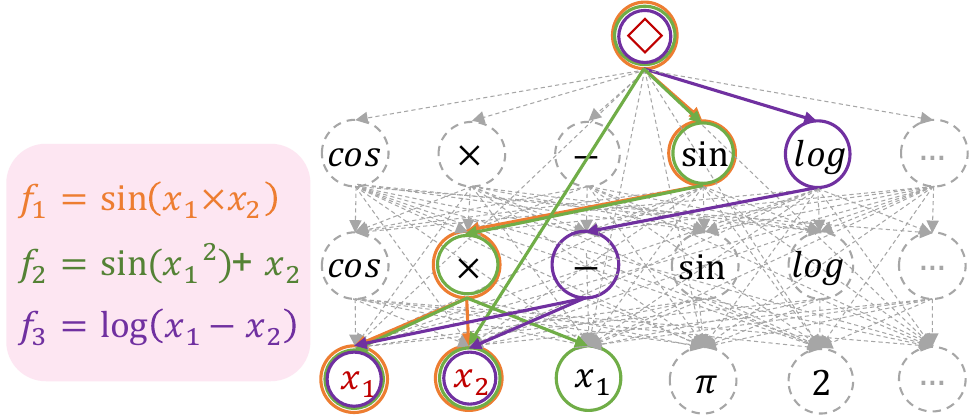}
\caption{Symbol Graph.}
\label{fig:steinergp}
\end{figure}
\section{Related Work}
In addressing the challenges of symbolic regression (SR), algorithms must navigate a vast mathematical expression space \citep{korns2013baseline}. While various approximate methods are employed to find results within this space, there is often a lack of explanation as to why approximate rather than deterministic methods are used. This typically stems from the perception that the SR problem is NP-hard \citep{lu2016using,virgolin2022symbolic,petersendeep,udrescu2020ai}, a notion that, until recently, had not been rigorously proven.

While \citeauthor{lu2016using} categorize the SR problem as NP-hard, they do not provide a formal proof. In contrast, \citeauthor{virgolin2022symbolic} offer a more concrete foundation for this classification by drawing parallels between the SR problem and the unbounded subset sum problem, effectively reducing the latter to an instance of the former. This approach provides a more solid theoretical basis for the use of approximate methods in solving SR problems, as it aligns the complexity of SR with the well-established NP-hard nature of the unbounded subset sum problem. The approach by~\citeauthor{virgolin2022symbolic} in proving the NP-hardness of the symbolic regression (SR) problem, though significant, comes with a notable limitation. Their method restricts the mathematical expression space only to include linear sums of features from the dataset. Specifically, they consider functions of the form $f(x)=c+\sum_{j=1}^d{x_jm_j}$ with $m_j \in N^+$ and $c\in R$, such as "$x+1$" and "$3x+5$". 

However, this constraint greatly simplifies the types of mathematical expressions that can be represented, failing to encompass a wider and more complex range of expressions frequently encountered in SR problems. For instance, expressions like "$sin(x^2+2x^3)$", "$sin(cos(x^2))$", and "$\sqrt{e^x}$" are examples of more intricate functions that are not covered under their defined expression space. These types of expressions are crucial in many SR applications, as they can represent a broader spectrum of real-world phenomena and mathematical relationships.
    
In contrast to the approach taken in~\cite{virgolin2022symbolic}, this paper introduces the concept of a symbol graph as a comprehensive representation of the entire mathematical expression space. Building upon this, we offer a robust proof demonstrating that the symbolic regression (SR) problem is equivalent to the degree-constrained Steiner Arborescence problem (DCSAP) \citep{chung1998algorithms,guo2022approximating}. This equivalence is significant as DCSAP is proven to be NP-hard in Lemma~\ref{dcsmtnpc}. Therefore, by establishing this, we conclusively prove that the SR problem is also NP-hard. This not only broadens our understanding of the SR problem's complexity but also validates the necessity of approximation methods in tackling it.

\section{The SR Problem is NP-hard}  \label{sec:srnp}
This paper establishes the NP-hardness of the symbolic regression (SR) problem through a three-step process. 

(1) The first step involves demonstrating that the Degree-Constrained Steiner Arborescence Problem (DCSAP) is NP-hard. DCSAP is a variant of the degree-constrained Steiner tree problem (DCSTP), which is already known to be NP-complete, thereby extending this complexity classification to DCSAP. 

(2) The second step connects the SR problem to DCSAP by illustrating that solving the SR problem is akin to finding a DCSAP in the symbol graph (introduced in Section~\ref{sec:graph}). This relationship is visually represented in Figure~\ref{fig:steinergp}. 

(3) The final step combines these insights to assert that since DCSAP is NP-hard, and the SR problem is equivalent to DCSAP, it logically follows that the SR problem is also NP-hard. 

\subsection{DCSAP is NP-hard}
The degree-constrained Steiner Arborescence problem (DCSAP) is a variant of the degree-constrained Steiner tree problem (DCSTP) that refers to finding a minimum-weight tree in an undirected graph \citep{voss1992problems,bauer1995degree,liu2003degree}. Let the undirected graph be $G=(V,E,w)$, where $V=\{ v_1,...,v_n\}$ is the vertex set, $E=\{e_{ij},... \}$ is the edge set, and $w$ is a weight function that maps an edge $e_{ij}$ into a real number, i.e., $w: E\to \mathbb{R}$. Given a vertex subset $S\subseteq V$ (called \textbf{terminals}), the DCSTP requires a minimal weighted tree $G_s=(V_s,E_s,w)$ with $S$ and some additional vertices in $G$ ($W(G_s)=\sum_{e\in E_s}{w(e)}$). For example, Figure~\ref{DCSTP} shows an undirected graph $G$, where the red circles are terminals. The DCSTP of $G$ connects in black lines, whose total weight is minimum among all possible choices of trees.
\begin{figure}[htp]
\centering
\subfigure[DCSTP]{
\centering
\includegraphics[width=0.39\linewidth]{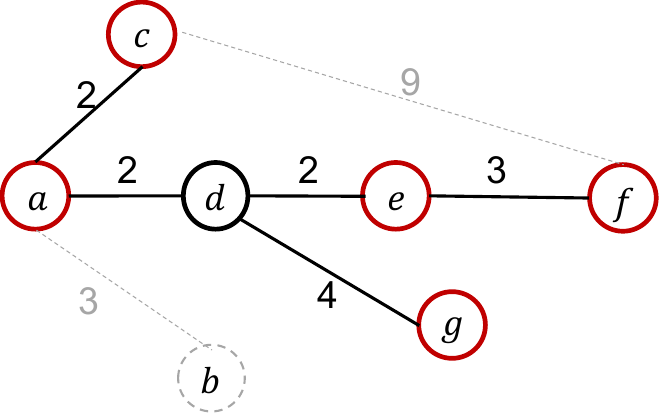}
\label{DCSTP}
}
\subfigure[DCSAP]{
\centering
\includegraphics[width=0.39\linewidth]{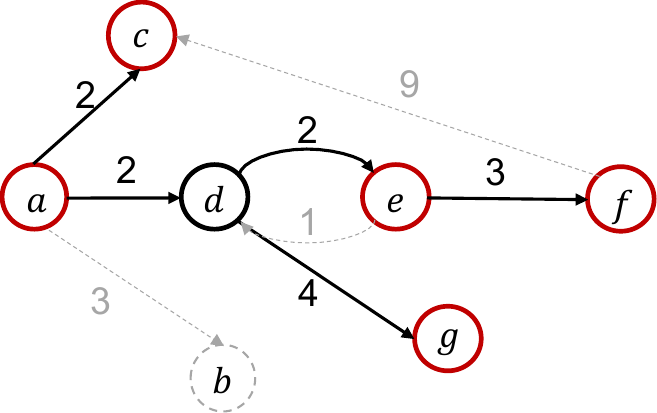}
\label{DCSAP}
}
\caption{An example of the DCSTP and DCSAP. (a) shows an undirected graph $G=(V,E,w)$ with the terminals $S$=\{$V_a,V_c,V_e,V_f,V_g$\}. The DCSTP is connected in black lines with a weight of 13. (b) shows a directed graph with the same terminals and the root vertex $r=V_a$.}
\label{fig:st}
\end{figure}

The DCSAP problem is the extension of the DCSTP problem, with the constraint on a directed graph. For example, Figure~\ref{DCSAP} shows a directed graph with a root vertex $r=V_a$. The tree connected by black lines is the DCSAP.
\begin{mylem}\label{dcsmtnpc}
The DCSAP problem is NP-hard.
\end{mylem}
\begin{proof} 
The DCSTP is NP-complete \citep{bauer1995degree,voss1990steiner}. Moreover, it can be easily known that any instance of the DCSTP in an undirected graph can be transformed into a DCSAP in a directed graph by replacing each edge $e(i,j)$ by two oppositely directed arcs $e\langle i,j\rangle$ and $e\langle j,i\rangle$ and then associating the weight of the edge $e(i,j)$ with them. So, the DCSAP is NP-hard. 
\end{proof} 
The DCSAP is an optimization problem. Its decision version (DCSAP-Dec) refers to deciding whether a DCSAP exists with a specific weight $\epsilon$ in the directed graph. It is defined formally as follows. 
\begin{myDef}[DCSAP-Dec]\label{DCSMT-Dec}
Given a directed graph $G=(V,E,w)$, DCSAP-Dec outputs a directed graph $G_s$ if and only if:
\begin{equation}
\exists \ e\in E_s \wedge d_i \leq k_i: W(G_s) =\sum_{e\in G_s} w(e) =\epsilon,
\end{equation}
where $d_i$ is the degree of the vertex $v_i$, $k_i$ is the degree constraint of $v_i$.
\end{myDef}
\begin{mylem}
The DCSAP-Dec is NP-complete.
\end{mylem}
\begin{proof}
Based on~\cite{papadimitriou2003computational}, to prove the DCSAP-Dec is NP-complete, we should prove (1) the DCSAP-Dec is in NP, and (2) the DCSAP-Dec is NP-hard.

(1) Demonstrating that DCSAP-Dec is in NP: This is achieved by verifying that a given solution $W(DCSAP)$ equals $\epsilon$ within polynomial time. The ability to perform this verification efficiently places DCSAP-Dec within the NP category.

(2) Proving that DCSAP-Dec is NP-hard: We prove the DCSAP-Dec is NP-hard since it can be transformed into the DCSAP via the bisection method. Let us assume that finding a solution for DCSAP-Dec $W(G_s)=\epsilon$ takes polynomial time in a directed graph (each edge's weight is $1$, and the number of vertex is $n$). Then we can use the bisection method to determine the minimum weight of DCSAP. The bisection method takes $O(log_2^n)$. Consequently, finding a DCSAP solution with the minimum weight would also require polynomial time. However, this contradicts the established understanding from Lemma~\ref{dcsmtnpc} that DCSAP is NP-hard. Therefore, our initial assumption must be incorrect, and DCSAP-Dec is also NP-hard.

The combination of these two points --- that DCSAP-Dec is in NP and is NP-hard --- leads to the conclusion that DCSAP-Dec is NP-complete. This conclusion is crucial as it underscores the computational complexity involved in solving the DCSAP-Dec and by extension, similar problems in computational theory.
\end{proof} 

\subsection{Symbol Graph} \label{sec:graph}
To describe the mathematical expression space, we construct a \textbf{symbol graph} $G=(V,E,W) $ as shown in Figure~\ref{fig:steinergp}. The symbol graph is a layered graph that includes three types of layers: 1) the top layer, 2) the leaf layer, and 3) the function layer. 
1) The top layer only contains a '$\diamond$' vertex as the root vertex. The '$\diamond$' vertex is the cumulative sum operator '$\sum$', which can linearly combine mathematical expressions from the second layer (function layer) and the leaf layer. For example, in Figure~\ref{fig:steinergp}, the '$\diamond$' vertex can represent $"sin(x_1\times x_2)+sin(x_1^2)+x_2+log(x_1-x_2)"$.
2) The leaf layer contains constant vertices $V_c$ and variable vertices $V_x$. Each constant vertex shows a constant, such as $1.2$, $e$, and $\pi$, and each variable vertex represents a variable. Note that $k$ variable vertices in $V_x$ may show the same variable. 
3) The function layer has $l$ levels. Each level consists of operator vertices $V_{op}$ that represent the mathmatical functions , such as '$+$', '$\times$', and '$sin$'. Moreover, on the same level, $p$ operator vertices in $V_{op}$ may show the same function. Each vertex $op$ in $V_{op}$ can connect all vertices of its next level and the leaf layer. 
So, the vertex set in the symbol graph is $V=\lbrace \diamond, V_{op}, V_x, V_c \rbrace$. 

The edge set in the symbol graph is $E=\lbrace E_c, E_x, E_{op}\rbrace$, where $E_c$ is the edge set from the operator vertex to a constant vertex, $E_x$ is the edge set from the operator vertex to a variable vertex, $E_{op}$ is the edge set from an operator vertex to an operator vertex. When connecting vertices to a tree, the degree of root '$\diamond$' vertex is lower than or equal to $s$, where $s$ is the number of vertices in the second layer and leaf layer; the out-degree of a leaf vertex is equal to 0; the out-degrees of a unary operator, binary operator, and ternary operator vertex are 1, 2 and 3 respectively. 

The weight $w$ on an edge $e \in E$ represents a vertex's output after excluding its parameters' influence. The following steps compute the weight $w$.

$\bullet$ If $e \in  \lbrace E_c, E_x \rbrace$, $w$ is the value of the constant or the variable vertex (Figure~\ref{weighta}). 

$\bullet$ If $e \in E_{op}$, $w$ is a set obtained by the following substeps. 1) $w$ is initialized to outputs of all possible mathematical expressions that the operator vertex $v$ shows. For example, in Figure~\ref{weightb}, $w = \lbrace a, b, c, a\times b, a\times c, b\times c \rbrace$. 2) $w$ subtracts the sum of weights that are parameters of the operator vertex $v$, such as "$f(a,b)-(a+b)$" and "$f(a)-a$". For example, in Figure~\ref{weightc}, $w = \lbrace a, b, c, a\times b-(a+b), a\times c-(a+c), b\times c-(b+c) \rbrace$.   

So, the sum of weights on edges in a mathematical expression $f$ equals its output, i.e., $W=\sum_{e\in f}{w(e)}=f(x)$. For example, for $"f=c\times a"$ in Figure~\ref{weightc}, $"c\times a = c\times a -(a+c) + a +c" $, where $"c\times a -(a+c)"$, '$a$' and '$c$' are weights in the '$f$'. 

% A vertex and its connected descendant vertexes form a subspace $\omega$. $\omega$ contains all mathematical expressions represented by trees with the vertex as the root, as shown in Figure~\ref{weightb}. Since the vertex's incoming edge $e$ records all its outputs, $e$ can show $\omega$. Therefore, \textbf{each edge is a mathematical expression subspace}.  
\begin{figure}[htp]
\centering
\subfigure[]{
\centering
\includegraphics[width=0.13\linewidth]{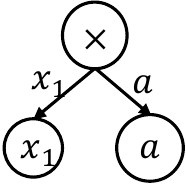}
\label{weighta}
}
\subfigure[]{
\centering
\includegraphics[width=0.22\linewidth]{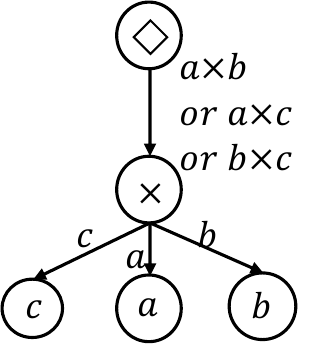}
 \label{weightb}
}
\subfigure[]{
\centering
\includegraphics[width=0.35\linewidth]{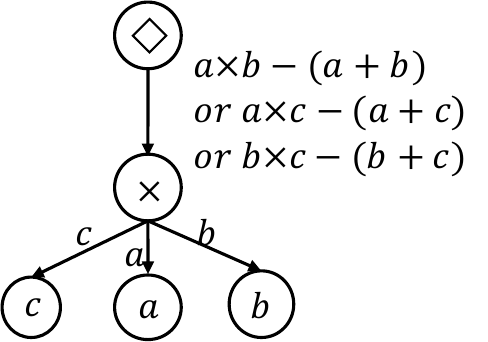}
 \label{weightc}
}
\caption{The examples of computing weights. ($a$) shows weights on edge $E_c$ and $E_x$; ($b$) and ($c$) show two substeps of calculating weights on edge $E_{op}$.}
\label{fig:weight}
\end{figure}

The symbol graph contains any mathematical expression by connecting a tree from the root '$\diamond$' vertex to vertexes in the leaf layer. For example, in Figure~\ref{fig:steinergp}, the tree connected with the green line represents the mathematical expression "$sin(x_1^2)+x_2$". Moreover, the symbol graph is the mathematical expression space $\Omega$. Finding a mathematical expression with an output of $Y$ can be seen as finding a Steiner Arborescence with the degree constraint and specified weights in the graph, which is the same as the DCSAP-Dec.
Therefore, the \textbf{decision version of SR problem (SR-Dec) equals DCSAP-Dec} in the symbol graph when setting the '$\diamond$' vertex as the root vertex, taking the root vertex and one variable vertex as \textbf{terminals}. 

Figure~\ref{fig:sr+st} shows an example of finding a DCSAP whose output is $Y$ of the given dataset after setting the three vertexes ('$\diamond$', '$x_1$' and '$x_2$') as terminals. The DCSAP is the mathematical expression $"sin(x_1\times x_2)"$ that is equal to the sum of weights in the tree, i.e.$"(sin(x_1 \times x_2) -x_1\times x_2)+(x_1\times x_2-x_1-x_2)+x_1+x_2"$.  
\begin{figure}[htp]
    \centering
    \includegraphics[scale=0.55]{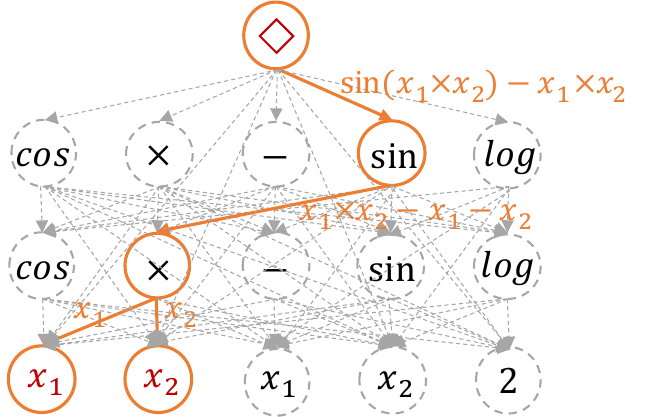}
    \caption{An example symbol graph G for the SR problem. The tree connected with orange lines is the DCSAP in G.}
    \label{fig:sr+st} 
\end{figure}

\subsection{The SR problem is NP-hard}
Here, we prove that the SR problem is NP-hard by showing that its decision version (SR-Dec) is NP-complete first.
The SR-Dec problem aims to find a function $f(X)$ in the mathematical expressions space $\Omega$ such that its loss $l$ is smaller than a chosen threshold $\epsilon$ \citep{virgolin2022symbolic}, i.e.,
\begin{equation}\label{equ2}
\exists f\in \Omega:l(Y,f(X)) \leq \epsilon.
\end{equation}
When $\epsilon$ equals 0, the SR-Dec is to find a function $f(X)$ with an output of $Y$.
\begin{myTheo}\label{srdecnpc}
The SR-Dec problem is NP-complete.
\end{myTheo}
\begin{proof}
Based on~\cite{papadimitriou2003computational}, to prove the SR-Dec is NP-complete, we should prove (1) the SR-Dec is in NP, and (2) the SR-Dec is NP-hard.

(1) Demonstrating that SR-Dec is in NP: This is achieved by verifying that a given solution $f(X)$'s loss $l\leq \epsilon$ within polynomial time. The ability to perform this verification efficiently places SR-Dec within the NP category.

(2) Proving that SR-Dec is NP-hard: we show that the SR-Dec problem is NP-hard since the SR-Dec problem equals the DCSAP-Dec problem. According to Section \ref{sec:graph}, the following holds:
\begin{align}
&\exists f\in \Omega:l(Y,f(X)) \leq \epsilon ?\\
(Setting \ \epsilon =0) \Rightarrow &\exists f\in \Omega:l(Y,f(X)) \leq 0 ?\\
\Rightarrow &\exists f\in \Omega:f(X)= Y ?\\
(W(G_s)=f(X))\Rightarrow &\exists e\in E_s \wedge d_i \leq k_i:W(G_s)= Y ?
\end{align}
Constructing a limited-level symbol graph to transform DCSAP-Dec into SR-Dec takes polynomial time. Since DCSAP-Dec is NP-complete, SR-Dec is NP-hard.

The combination of these two points --- that SR-Dec is in NP and is NP-hard --- leads to the conclusion that SR-Dec is NP-complete. This conclusion is crucial as it underscores the computational complexity involved in solving the SR-Dec and by extension, similar problems in computational theory. 
\end{proof}
According to~\cite{virgolin2022symbolic}, it can be concluded that since the SR-Dec problem is NP-complete, then the SR problem is NP-hard.

\section{Conclusion}
In this paper, we design a novel symbol graph that can describe the mathematical space of the symbolic regression (SR). Based on the symbol graph, we offer a robust proof demonstrating that the SR problem is equivalent to the degree-constrained Steiner Arborescence problem (DCSAP). The complexity of DCSAP, which is proven to be NP-hard in Lemma~\ref{dcsmtnpc}, directly implies the SR problem is also NP-hard. Our proof comprehensively covers all mathematical expressions, going beyond the simpler expressions considered in the previous proof.

\bibliographystyle{unsrtnat}
\bibliography{ref.bib}

\end{document}